\documentclass{article}
\usepackage{amsmath,amsthm, amssymb, latexsym, enumerate, graphicx, verbatim, subfig, caption, bbold}
\usepackage[square,sort,comma,numbers]{natbib}
\usepackage{authblk}
\usepackage{tikz}
\usepackage{enumerate}

\newtheorem{thm}{Theorem}[section]

\theoremstyle{remark}

\theoremstyle{definition}
\newtheorem{defin}[thm]{Definition}

\newtheorem{experiment}[thm]{Experiment}

\newcommand{\vect}[1]{\boldsymbol{#1}}
\newcommand{\dotprod}[2]{\left \langle {#1}, {#2} \right \rangle}

\newcommand{\E}[1]{E\!\left [ {#1} \right ]}
\newcommand\numberthis{\addtocounter{equation}{1}\tag{\theequation}}
\newcommand{\ktimes}{\overset{\scriptscriptstyle K} \otimes}
\newcommand{\R}{{\mathbb R}}

\newcommand{\mspan}{\operatorname{span}}
\newcommand{\diag}[1]{\operatorname{diag}\!\left({#1}\right)}
\newcommand{\1}{\mathbb{1}}
\newcommand{\proj}{\operatorname{Proj}}
\title{An Online Algorithm for Learning Selectivity to Mixture Means}
\author[1]{Matthew Lawlor}
\author[2]{Steven Zucker}
\affil[1]{Yale University\footnote{now at Google Inc.}, matthew.lawlor@yale.edu}
\affil[2]{Yale University, zucker@cs.yale.edu}

\begin{document}
\maketitle

\begin{abstract}
We develop a biologically-plausible learning rule called Triplet BCM that provably
converges to the class means of general mixture models.  This rule generalizes the
classical BCM neural rule, and provides a novel interpretation of classical BCM as performing a kind of tensor decomposition.  It achieves a substantial generalization over classical BCM by incorporating triplets of samples from the mixtures, which provides a novel information processing interpretation to spike-timing-dependent plasticity.  We provide complete proofs of 
convergence of this learning rule, and an extended discussion of the connection between BCM and tensor learning.    
\end{abstract}

Spectral tensor methods are emerging themes in machine learning, but they
remain global rather than ``on-line.'' While incremental (on-line) learning can be useful in many practical applications, it is essential for biological learning. We introduce a triplet learning rule for mixture distributions based on a tensor formulation of the BCM biological learning rule. It is implemented in a feed forward fashion, removing the need for backpropagation of error signals.  

Our main result is that a modified version of the classical Bienenstock-Cooper-Munro \cite{bienenstock1982theory} synaptic update rule, a neuron can perform a tensor decomposition of the input data.  By incorporating the interactions between input triplets (commonly referred to as a multi-view assumption), our learning rule can provably learn the mixture means under an extremely broad class of mixture distributions and noise models. This improves on the classical BCM learning rule, which will not converge properly in the presence of
noise. We also provide new theoretical interpretations of the classical BCM rule, specifically we show the classical BCM neuron objective function is closely related to some objective functions in the tensor decomposition literature, when the input data consists of discrete input vectors.  We also prove convergence for our modified rule when the data is drawn from a general mixture model.

The multiview requirement has an intriguing
implication for neuroscience. Since spikes arrive in waves, and 
spike trains matter for learning \cite{froemke2002spike},  our model suggests that
{\em the waves of spikes arriving during adjacent epochs in time provide multiple
samples of a given stimulus}. This provides a powerful information processing interpretation to biological learning.
To realize it fully, we note that while classical BCM can be implemented via
spike timing dependent plasticity \cite{pfister2006triplets}\cite{gjorgjieva2011triplet}\cite{caporale2008spike}\cite{song2000competitive}. However, 
most of these approaches require much stronger distributional assumptions on the input data, or learn a much simpler decomposition of the data than our algorithm.  Other, Bayesian methods \cite{nessler2009stdp},  require the computation of a posterior distribution with implausible normalization requirements.  Our learning rule successfully avoids these issues, and has provable guarantees of convergence to the true mixture means. 

This article forms an extended technical presentation of 
some proofs introduced at NIPS 2014\cite{lawlor2014feedforward}, which has more discussion
on the implications for biological learning, as well as fits of this model to spike timing dependent plasticity data.
We will not formalize the connection to biology in this article, instead we present a connection between classical BCM and tensor decompositions, and a proof that under a broad class of mixture models the triplet BCM rule can learn selectivity to a single mixture.  We also show that a laterally connected network of triplet BCM neurons will each learn selectivity to different components of the mixture model. 

The outline for this article is as follows:
\begin{itemize}
  	\item Tensor notation and tensor decomposition of mixture moments under the triplet input model
  	\item Introduction to classical BCM
  	\item Connection between classical BCM and tensor decompositions
  	\item Definition of triplet BCM, and proof of convergence of \emph{expected update} under the triplet input model
  	\item Finally, the main contribution of this article is a proof of convergence with probability one under the triplet input model. 
  \end{itemize}  

\section{Notation for Tensor Products}
Following Anandkumar et. al., \cite{anandkumar2012tensor} we will use the following notation for 
tensors.  Let $\otimes$ denote the tensor product.  If $T = \vect{v}_1 \otimes ... \otimes \vect{v}_k$ then we say that \[T_{i_1, ..., i_k} = \prod_{j = 1}^k \vect{v}_j(i_j) \]

We denote the application of a $k$-tensor to $k$ vectors by $T(\vect{w}_1, ..., \vect{w}_k)$
where 
\[T(\vect{w}_1, ..., \vect{w}_k) = \sum_{i_1, ..., i_k} T_{i_1, ..., i_k}\prod_j\vect{w}_j(i_j)\] so in the simple case where $T = \vect{v}_1 \otimes ... \otimes \vect{v}_k$, 
\[ T(\vect{w}_1, ..., \vect{w}_k) = \prod_j \dotprod{\vect{v}_j}{\vect{w}_j}\]

We further denote the application of a $k$-tensor to $k$ matrices by $T(M_1, ..., M_k)$
where 
\[T(M_1, ..., M_k)_{i_1, ..., i_k} = \sum_{j_1, ..., j_k} T_{j_1, ..., j_k}[M_1]_{j_1, i_1}...[M_k]_{j_k,i_k}\]

Thus if $T$ is a $2$-tensor, $T(M_1, M_2) = M_1^T T M_2$ with ordinary matrix multiplication.
Similarly, $T(\vect{v}_1, \vect{v}_2) = \vect{v}_1^T T \vect{v}_2$

We say that $T$ has an orthogonal tensor decomposition if 

\[T = \sum_k \vect{v}_k \otimes \vect{v}_k \otimes ... \otimes \vect{v}_k \] and
\[\dotprod{\vect{v}_i}{\vect{v}_j} = \delta_{i}^j\]

For more on orthogonal tensor decompositions see \cite{anandkumar2012tensor}.  Let $T = \sum_k \lambda_k \mu_k \otimes \mu_k \otimes \mu_k$ and $M = \sum_k \lambda_k \mu_k \otimes \mu_k$ where $\mu_k \in \R^n$ are assumed to be linearly independent, and $\lambda_k > 0$.  We also assume , $n\ge k$, so $M$ is a symmetric, positive semidefinite, rank $k$ matrix.  
Let $M = UDU^T$ where $U\in \mathbb{R}^{n \times k}$ is unitary and $D \in R^{k \times k}$ is diagonal.  Denote $W = UD^{-\frac{1}{2}}$.  Then
$M(W, W) = I_k$.  Let $\tilde{\mu}_k = \sqrt{\lambda_k} W^T \mu_k$.  Then
\begin{align}
M(W, W) = W^T \sum_k \sqrt{\lambda_k}\mu_k \otimes \sqrt{\lambda_k} \mu_k W = \sum_k \tilde{\mu}_k \tilde{\mu}_k^T = I_k \label{OrthTensor}
\end{align}
Therefore $\tilde{\mu}_k$ form an orthonormal basis for $\mathbb{R}^k$.  Let 
\begin{align*}
	\tilde{T} &= T(W, W, W) \\
	&= \sum_k \lambda_k (W^T \mu_k) \otimes (W^T \mu_k) \otimes (W^T \mu_k) \\
	&= \sum_k \lambda_k^{-\frac{1}{2}} \tilde{\mu}_k \otimes \tilde{\mu}_k \otimes \tilde{\mu}_k 
	\numberthis \label{OrthTensor2}
\end{align*}
We say $\tilde{T}$ is an orthogonal tensor of rank $k$.  

\subsection{Tensors and Mixture Models}
With the notation for tensors established, we return to moments of mixture models under our assumptions. 

Let 
\begin{equation}
P(\vect{d}) = \sum_{k = 1}^K \alpha_k P_k(\vect{d})
\end{equation}
$\vect{d} \in \R^n$, $k \le n$.  We will denote data vectors $\vect{d}$ drawn independently from the same
conditional distribution $P_k$ with superscripts.  For example, $\{\vect{d}^1,  \vect{d}^2, \vect{d}^3\}$ denotes a triple drawn from \emph{one} of
$\{P_1, \ldots, P_k\}$.  To emphasize the triplet input model, we point out that while the marginal distribution of any of $\{\vect{d}^1,  \vect{d}^2, \vect{d}^3\}$ is
\begin{equation}
P(\vect{d^i}) = \sum_{k = 1}^K \alpha_k P_k(\vect{d^i})
\end{equation}
the joint distribution of $\{\vect{d}^1,  \vect{d}^2, \vect{d}^3\}$ is \emph{not} the product
of these marginal distributions.  For the following equations, all expectations containing superscripts are taken with respect to the triplet distribution, and all equations without are taken with respect to the marginal distribution, or independent products of it depending on context.  Let 
\begin{equation}
E_{P_k}[\vect{d}] = \vect{d}_k
\end{equation}
Then, by the conditional independence of $\vect{d}_1, \vect{d}_2, \vect{d}_3$ 
\begin{align*} 
E[\vect{d}] &= \sum_k \alpha_k \vect{d}_k \\
E[\vect{d}^1 \otimes \vect{d}^2] &= \sum_k \alpha_k \vect{d}_k \otimes \vect{d}_k \\
E[\vect{d}^1 \otimes \vect{d}^2 \otimes \vect{d}^3] &= \sum_k \alpha_k \vect{d}_k \otimes \vect{d}_k \otimes \vect{d}_k
\end{align*}

These estimators are in the spirit of classical method of moment estimators.  Classical 
method of moment estimators try to write the parameter to be estimated as a function of
moments of the distribution.  The moments are then plugged into the resulting equations.
Here, a decomposition of a moment tensor is used as an estimator for the desired parameters.

To give an indication of the importance of the multi-view assumption, we note that with only access to vectors drawn independently from the full distribution, we would be restricted to moments like the following:
\begin{align*}
E[\vect{d} \otimes \vect{d}] &= \sum_k \alpha_k \vect{d} \otimes \vect{d} + D
\end{align*}
where $D_{ii} = \sum_k \alpha_k (E_{P_k}[d_{ki}^2] - E_{P_k}[d_{ki}]^2)$

The diagonal matrix $D$ ensures that the moment matrix is not low rank.  A similar phenomenon occurs for the third-order tensor.  For some classes of mixture distributions, low-rank moment
tensors can be constructed even without multiple samples from  the mixture components.  However,
these methods require specific structure to the mixture components, and do not generalize to all mixture distributions.  

Classial methods for fitting mixture models, like EM tend not to have formal guarentees of convergence to a global optimum.  In general, optimum fitting of mixture models is believed to 
be quite hard under many circumstances \cite{arora2005learning}.  The two assumptions we require,
that the mixture means span a low rank subspace, and that we have access to three samples known to come from the same latent class, allow us to skirt these difficulties.

When this structure exists, the approach of \cite{anandkumar2012tensor}  is to try to find a low rank decomposition of these tensors.  Unfortunately,
storing and then decomposing these tensors is not an option under our biological restrictions.   We now turn to the most significant technical contribution of this article: a biologically plausible \emph{online} learning algorithm for learning selectivity to individual mixture components 
under a mixture model.  Typical proofs of convergence for tensor mixture methods tend to first use a central limit argument to show convergence of the moments.  Then, they show that for an orthogonal tensor with small errors, the errors in the orthogonal decomposition will also be small.  We do not explicitly compute these moments, and instead show that our online algorithm will converge with probability one through a stochastic optimization argument.  

We show that not only can selectivity to mixture be learned, but that the algorithm also provides a new interpretation for 
sequences of action potentials: disjoint spiking intervals provide multiple views of a distribution.

\section{Introduction to BCM}
 
The original formulation of the BCM rule is as follows: Let $c$ be the post-synaptic firing rate,
$\vect{d}\in{\R^N}$ be the vector of presynaptic firing rates, and $\vect{m}$ be the vector of synaptic
weights.  Then the BCM synaptic modification rule is
\begin{align*}
c &= \dotprod{\vect{m}}{\vect{d}} \\
\dot{\vect{m}} &= \phi(c, \theta)\vect{d}
\end{align*}
$\phi$ is a non-linear function of the firing rate, and $\theta$ is a sliding threshold that increases as a superlinear function of the average firing rate. 

There are many different formulations of the BCM rule.  The primary 
features that are required are :
\begin{enumerate}
	\item $\phi(c, \theta)$ is convex in $c$
	\item $\phi(0,\theta) = 0$
	\item $\phi(\theta, \theta) = 0$
	\item $\theta$ is a super-linear function of $E[c]$ 
\end{enumerate}
\begin{figure}[ht]
	\centering
	\includegraphics[width=.75\textwidth]{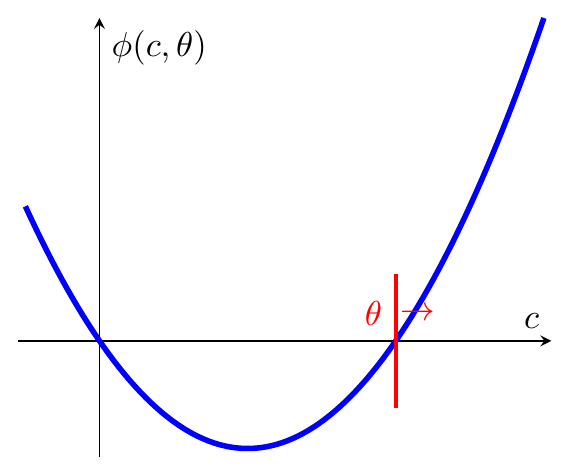}
	\caption[Sliding threshold for BCM rule]{BCM rule.  $\theta$ is a sliding threshold
	which is superlinear in $c$.}
\end{figure}
These properties guarantee that the BCM learning rule will not grow without bound.  There
have been many variants of this rule.  One of the most theoretically well analyzed 
variants is the Intrator and Cooper model \cite{intrator1992objective}, which
 has the following form for 
$\phi$ and $\theta$.
\begin{align*}
	\phi(c, \theta) &= c(c-\theta) \text{ with }	\theta = E[c^2]
\end{align*}

\begin{defin}[BCM Update Rule] For the purpose of this article, the BCM rule is defined as
\begin{equation}
\vect m_n = \vect m_{n-1} + \gamma_n c_n(c_n - \theta_{n-1})\vect{d}_n \label{BCM Update Rule}
\end{equation}

where $c_n = \dotprod{\vect m_{n-1}}{\vect d_n}$ and $\theta = E[c^2]$.  $\gamma_n$ is a sequence
of positive step sizes with the property that $\sum_n \gamma \rightarrow \infty$ and 
$\sum_n \gamma_n^2 < \infty$
\end{defin}

The traditional application of this rule is a system where the input $\vect{d}$ is drawn from 
linearly independent vectors
$\{\vect{d}_1, ..., \vect{d}_K\}$ with probabilities $\alpha_1, ..., \alpha_K$, with $K = N$, the dimension of the space.  

These choices are quite convenient because they lead to the following objective function 
formulation of the synaptic update rule.

\[R(\vect m) = \frac{1}{3} \E{\dotprod{\vect{m}}{\vect{d}}^3} -
\frac{1}{4} \E{\dotprod{\vect{m}}{\vect{d}}^2}^2\]

Thus,
\begin{align*}
\nabla R &= \E{\dotprod{\vect{m}}{\vect{d}}^2 \vect{d} - E[\dotprod{\vect{m}}{\vect{d}}^2] \dotprod{\vect{m}}{\vect{d}}\vect{d}} \\
&= c(c - \theta)\vect d \\
&= \phi(c, \theta) \vect d
\end{align*}

So in expectation, the BCM rule performs a stochastic gradient ascent in $R(\vect{m})$.  

With this model, we observe that the objective function can be rewritten in tensor notation.  Note that this input model can be seen as a kind of degenerate 
mixture model.

This objective function can be written as a tensor objective function, by noting the 
following:  
\begin{align}
	T &= \sum_k \alpha_k \vect{d}_k \otimes \vect{d}_k \otimes \vect{d}_k \nonumber\\
	M &= \sum_k \alpha_k \vect{d}_k \otimes \vect{d}_k \nonumber\\
	R(m) &= \frac{1}{3} T(\vect{m}, \vect{m}, \vect{m}) - \frac{1}{4}M(\vect{m}, \vect{m})^2 \label{BCM Objective Function}
\end{align}
Building off of the work of \cite{anandkumar2012tensor} we will use this characterization
of the objective function to build a triplet BCM update rule which will converge for 
general mixtures, not just degenerate ones.  

For completeness, we present a proof that the stable points of the expected BCM
update are selective for only one of the data vectors.

The stable points of the expected update occur when $E[\dot{\vect{m}}] = 0$.  Let 
$c_i = \dotprod{\vect{m}}{\vect{d}_i}$ and $\phi_i = \phi(c_i, \theta)$.  
Let $\vect{c} = [c_1, \ldots, c_K]^T$ and $\Phi = [\phi_1, \ldots, \phi_K]^T$.  
\begin{align*}
D^T &= \begin{bmatrix}
	\vect{d}_1 & | & \cdots & |  & \vect{d}_k
	\end{bmatrix} \\
P &= \diag{\vect{\alpha}}
\end{align*}
\begin{thm} \label{BCMStablePoints}
	(Intrator 1992) Let $K = N$, linearly independent $\vect d_k$, and let  $\alpha_i>0$ and distinct.   Then stable points (in the sense of Lyapunov) of the expected update 
	$\dot{\vect{m}} = \nabla R $ occur when $\vect{c} = \alpha_i^{-1} e_i$ or 
	$\vect{m} = \alpha_i^{-1} D^{-1} e_i$ 
\end{thm}
\begin{proof}
$E[\dot{\vect{m}}] = D^T P \Phi$ which is 0 only when $\Phi = 0$.  Note 
$\theta = \sum_k \alpha_k c_k^2$.  $\phi_i = 0$ if $c_i = 0$ or $c_i = \theta$.  
Let $S_+ = \{i:c_i \ne 0\}$, and $S_- = \{i:c_i = 0\}$.  Then for all $i \in S_+$, $c_i = \beta_{S_+}$

\begin{align*}
\beta_{S_+} - \beta_{S_+}^2\sum_{i \in S_+} \alpha_i &= 0 \\
\beta_{S_+} &= \left(\sum_{i \in S_+} \alpha_i \right)^{-1} \\
\end{align*}

Therefore the solutions of the BCM learning rule are $c = \1_{S_+} \beta_{S_+}$, for all
subsets $S_+ \subset \{1, \ldots, K\}$.  We
now need to check which solutions are stable.  The stable points (in the sense of Lyapunov) 
are points where the matrix 
\begin{equation*}
	H = \frac{\partial E[\dot{\vect{m}}]}{\partial \vect{m}} 
\end{equation*}
is negative semidefinite.  

\begin{align}
H &= D^T P \left( \frac{\partial \Phi}{\partial \vect{c}}\right) \frac{\partial \vect{c}}{\partial \vect{m}} \nonumber \\
&= D^T P \left( \frac{\partial \Phi}{\partial \vect{c}}\right) D \label{stability}
\end{align}
Let $S$ be an index set $S \subset \{1, \ldots, n\}$.  We will use the following notation
for the diagonal matrix $I_S$:
\begin{align}
(I_S)_{ii} = \begin{cases} 1 & i\in S \\ 0 & i \notin S \end{cases} \label{Definition of IS}
\end{align}
So $I_{S} + I_{S^c} = I$, and $\vect e_i \vect e_i^T = I_{\{i\}}$

a quick calculation shows 
\begin{align*}
\left( \frac{\partial \phi_i}{\partial c_j} \right) &= \beta_{S_+}I_{S_+} - \beta_{S_+} I_{S_-} -2 \beta_{S_+}^2 \diag{\vect{\alpha}}\1_{S_+}\1_{S_+}^T
\end{align*}
This is negative semidefinite iff $A = I_{S_+} - 2 \beta_{S_+} \diag{\vect{\alpha}} \1_{S_+}\1_{S_+}^T$ is negative semidefinite.  

Assuming a non-degeneracy of the probabilities $\vect{\alpha}$, and assume $|S_+| > 1$.  Let $j = \arg \min_{i \in S_+} \alpha_i$.  Then $\beta_{S_+} \alpha_j < \frac{1}{2}$ so $A$ is not negative semi-definite.  However, if  $|S_+| = 1$ then
$A = -I_{S_+}$ so the stable points occur when $\vect c = \frac{1}{\alpha_i} \vect e_i$
\end{proof}
Each stable state is selective for only one $\vect{d}_i$.  We now show a connection 
between this objective function, and another objective function from the tensor
decomposition literature.  

\section{Connection Between BCM and a Generalized Eigen-Tensor Decomposition}
We now briefly show a connection between a BCM-like learning rule and a generalized eigen-tensor decomposition.  We will show
that the stable points of this learning rule are, up to a constant power of the weights,
identical to the stable points of the BCM rule.  

Rather than a sliding threshold that penalizes the activity of the neuron, we modify the BCM
neuron with a sliding threshold that drives the expected activity of the neuron to a \emph{specified} 
activity level.  This will allow us to rewrite the objective function of the neuron as a generalized tensor spectral decomposition.  

Let \begin{align*}
	\hat{R}(\vect{m}, r) &= \frac{1}{3}T(\vect{m}, \vect{m}, \vect{m}) + \frac{r}{3} \left(1 - 
	M(\vect{m}, \vect{m})^2 \right)
\end{align*}

Let $W$ be defined as in equation \ref{OrthTensor}.  Let $\vect{m} = W\vect{u}$.  Then 
\begin{align}
\hat{R}(W\vect{u}, r) &= \frac{1}{3}T(W \vect{u}, W \vect{u}, W \vect{u}) + \frac{r}{3} M(\vect{m}, \vect{m})^2 \nonumber\\
&= \frac{1}{3} \tilde{T}(\vect{u},\vect{u},\vect{u}) + \frac{r}{3} (1 - \dotprod{\vect{u}}{\vect{u}}^2) \label{u definition}
\end{align}
where $\tilde{T}$ is defined as in equation \eqref{OrthTensor2}.  We note that $\tilde{T}$ is an orthogonal tensor.  This equation the Lagrange multiplier formulation of a generalized tensor spectral expansion for an orthogonal tensor
with $r$ as a Lagrange multiplier.
The analogy with the general eigenvector expansion is as follows:  The first eigenvector
of a symmetric matrix $M$ is the solution to the following objective function:
\[\max_{\dotprod{\vect{u}}{\vect{u} }^2 = 1}M(\vect{u},\vect{u}) \] 
Our objective function attempts to find \[\max_{\dotprod{\vect{u}}{\vect{u} }^2 = 1}T(\vect{u},\vect{u}, \vect{u})\]
Unlike the symmetric matrix case, which has a single local maximum (assuming no degeneracy in the eigenvalues), the tensor objective function has many local
maxima.  For the matrix case, one can find additional eigenvectors by \emph{deflating} the 
matrix.  The process works by looking at a sequence of matrices created by 
successively subtracting out the low-rank matrix approximations generated by the eigenvectors, then repeating.  
\[M' = M - \beta \vect{u}\vect{u}^T \text{  where  }M\vect{u} = \beta \vect{u}\]
A similar approach can work with the orthogonal tensors.  However for orthogonal tensors, the local optima of the tensor objective function correspond
to global optima of some stage of the deflation process.  We do not need to explicitly deflate
the tensor to find its decomposition, we just need to ensure that each version of our gradient ascent ends at
a different local maximum.  A parallel algorithm using a network of neurons which can perform
this simultaneous search is presented in section \ref{Network}

With this objective function, the expected update rule becomes

\begin{align}
E[\dot{\vect{m}}] &= E[ \nabla R(\vect{m}, r)] \nonumber\\
&= E[\hat{\phi}(c, r\theta)\vect{d}] \nonumber \\
E[\dot{r}] &= -(1 - \theta^2) \label{eigentensor}
\end{align}
where $\hat{\phi} = c(c - \frac{4}{3}r \theta)$

\begin{defin}[Tensor BCM] The Tensor BCM learning rule is given by
\begin{align}
\vect{m}_n &= \vect{m}_{n-1} + \gamma_n \hat{\phi}(c_n, \theta_{n-1}r)\vect d_n \label{Tensor BCM} \\
r_n &= r_{n-1} - \gamma'_n (1 - \theta_{n-1}^2)
\end{align}

\end{defin}
\begin{thm}
	The constrained local maxima of \eqref{eigentensor} are \[\vect{m} = \lambda_k^{\frac{1}{2}} M^{-1}\mu_k\] \label{ConstrainedAlgorithm}
\end{thm}

Thus the modified BCM neuron learns decorrelated versions of the parameter vectors $\mu_k$.  In 
contrast with the ordinary matrix (2-tensor) eigendecomposition, this update function can converge to each of the eigenvectors of the 3-tensor, rather than just the one corresponding to the largest eigenvalue.  

\section{Triplet BCM Learns Selectivity to Components of Mixture Models} \label{Triplet BCM Chapter}
We have seen in the previous section that the classical BCM rule can be written as 
stochastic gradient ascent in a tensor objective function, provided the input consists of $N$ discrete vectors, where $N$ is the dimension of the input data.  We demonstrate 
that, under a multi-view and low rank assumption, this rule can be modified to learn selectivity for
mixture means under a broad variety of mixture models.  A neuron modifying its synaptic
selectivity under this rule would have positive expectation of firing for only one 
mixture component.  We will call this learning rule triplet BCM, as it requires access to 
triplets of data from each mixture distribution.  

First, we will describe the triplet BCM learning
rule, and show that the expected update of this rule with mixture model input will converge to a state which is 
selective for one and only one mixture.  Second, we will show that for a variety of update step sizes, this algorithm will converge w.p. 1 to the stable states of the expected
update.  Finally, we will show how to combine a network of these triplet BCM neurons to perform 
a parallel search of the mixture components, ensuring that each neuron is selective for one 
mixture component.  

\section{Triplet BCM Rule}
We now show that by modifying the update rule to incorporate information from triplets of 
input vectors, the generality of the input data can be dramatically increased.  Assume that 
\[P(\vect{d}) = \sum_k \alpha_k P_k(\vect{d})\]
where $E_{P_k}[\vect{d}] = \vect{d}_k$.  For example, the data could be a mixture of axis-aligned Gaussians, a mixture of independent Poisson 
variables, or mixtures of independent Bernoulli random variables to name a few.  We also require $E_{P_k}[\|\vect{d}\|^2] < \infty$.  We emphasize that 
we do not require our data to come from any parametric distribution.  

We interpret $k$ to be a latent variable that signals the hidden cause of the 
underlying input distribution, with distribution $P_k$.  Critically, we assume that the hidden 
variable $k$ changes slowly compared to the inter-spike period of the neuron.  In particular, 
we need at least 3 samples from each $P_k$.  This corresponds to the multi-view assumption of \cite{anandkumar2012tensor}.  A particularly relevant model meeting this assumption is that of spike counts in disjoint 
intervals under a Poisson process, with a discrete, time varying rate parameter.  

Let $\{\vect{d}_k^1, \vect{d}_k^2, \vect{d}_k^3\}$ be a triplet of independent copies from some $P_k(\vect{d})$, 
i.e. each are drawn from the same latent class.  It is critical to note that if $\{\vect{d}_k^1, \vect{d}_k^2, \vect{d}_k^3\}$ are not drawn from the same class, this update will not converge to the global maximum.  It appears that this assumption can be violated somewhat in practice with limited change to the fixed point of the algorithm, however we will not explore this further in this article.  Our sample is then a sequence of triplets, each triplet drawn from the same latent distribution.  Let $c_k^i = \dotprod{\vect{d}^i}{\vect{m}}$.  With these independent triples, we note that the tensors $T$ and $M$ from equation \eqref{BCM Objective Function} can be written as moments of the independent triplets
\begin{align*}
	T &= E[\vect d^1 \otimes \vect d^2 \otimes \vect d^3]\\
	M &= E[\vect d^1 \otimes \vect d^2] \\
	R(m) &= \frac{1}{3} T(\vect{m}, \vect{m}, \vect{m}) - \frac{1}{4}M(\vect{m}, \vect{m})^2
\end{align*}
As with classical BCM, we can perform gradient ascent in this objective function which leads 
to the expected update

\[E[\nabla R] = E[ c^1 c^2 \vect d^3 + (c^1 \vect d^2 + c^2 \vect d^1)(c^3 -2 \theta)]\]
where $\theta = E[c^1c^2]$.  This update is rather complicated, and couples pre and post
synaptic firing rates across multiple time intervals.  Since each $c^i$ and $\vect d^i$
are identically distributed, this expectation is equal to 
\begin{equation*}
E[c^2(c^3 - \theta)\vect d^1] \label{Triplet BCM Update}
\end{equation*}
which suggests a much simpler update.  This ordering was chosen to  match the 
spike timing dependency of synaptic modification.
\begin{defin}[Full-rank Triplet BCM]

We define the full-rank Triplet BCM update rule as:
\begin{equation}
\vect m_n = \pi(\vect m_{n-1} + \gamma_n \phi(c^2, c^3, \theta_{n-1}) \vect d^1 )\label{Full Rank Triplet BCM}
\end{equation}
where $\phi(c^2, c^3, \theta) = c^2(c^3 - \theta)$, $\sum_n \gamma_n \rightarrow \infty$, and
$\sum_n \gamma_n^2 < \infty$.  $\pi$ is a projection into the set $B_r := \{\vect{m}: \vect{m}^T M \vect{m} < r\}$ for a very large $r$.  See subsection \ref{Full Rank proof} for details on the projection.  
\end{defin}
\begin{defin}[Low-rank Triplet BCM]

The low-rank Triplet BCM update rule is:
\begin{equation}
\vect m_n = \pi((1 - \delta_n)\vect m_{n-1} + \gamma_n \phi(c^1, c^3, \theta_{n-1}) \vect d^2) \label{Low Rank Triplet BCM}
\end{equation}
where 
\begin{align*}\phi(c^2, c^3, \theta) &= c^2(c^3 - \theta)\\
\sum_n \gamma_n \rightarrow \infty  \text { \hspace{2ex} }& \text{ \hspace{2ex}} \sum_n \delta_n \gamma_n \rightarrow \infty\\
 \sum_n \delta_n^2 \gamma_n^2 < \infty \text{ \hspace{2ex}}&\text{ \hspace{2ex}} \sum_n \gamma_n^2 < \infty\\
 \delta_n &\rightarrow 0 \\
\end{align*}
$\pi$ is a projection onto the ball $B_r := \{\vect{m}: \vect{m}^T \hat{M} \vect{m} < r\}$ where $\hat{M} = M + \proj_{W^c}$, and $r$ is large.  See subsection \ref{Low Rank Proof} for details on the projection.
\end{defin}
 In practice, for 
a sufficiently small step size the projections rarely occur, and can be made arbitrarily infrequent through a sufficiently large choice of $r$.    	
\begin{thm} \label{Expected Update Triplet BCM}
Let $K = N$, and $\{\vect d_n^1, \vect d_n^2, \vect d_n^3 \}$ be multi-view triplets drawn from
a mixture model $P(\vect d) = \sum_n^K P_k(\vect d) $ with linearly independent means and bounded variance. 
With multi-view triplets drawn from a mixture model, the full-rank triplet BCM rule
will converge w.p. 1 and 
\[E[\dotprod{\vect m}{\vect d}] = \alpha_i^{-1} \vect{e}_i\]
\end{thm}
\begin{proof}
	The expected update for triplet BCM under a mixture model is identical to classical BCM
	with discrete data.  Therefore the proof of Theorem \eqref{BCMStablePoints} for the
	stable points of the expected update goes through unchanged.  The proof of convergence
	w.p. 1 for the true update requires some additional machinery, which will be covered in section \ref{Stochastic Approximation}.
\end{proof}
\begin{thm}
Let $K \le n$, and $\{\vect d_n^1, \vect d_n^2, \vect d_n^3 \}$ be multi-view triplets drawn from
a mixture model $P(\vect d) = \sum_n^K P_k(\vect d) $ with linearly independent means. 
With multi-view triplets drawn from a mixture model, the low-rank triplet BCM rule
will converge w.p. 1 and
\[E[\dotprod{\vect m}{\vect d}] = \alpha_i^{-1} \vect{e}_i\]
\end{thm}
\begin{proof}
See section \ref{Low Rank Proof}
\end{proof}

In expectation, each stable point is selective for one and only one mixture.  That
is to say, $E_{P_k}[\dotprod{\vect d}{ \vect m}]$ is non-zero for only one $k$.  This does not
preclude $E_{P_k}[c^2] $ from being quite large relative to $E_{P_k}[c]$.  In the case where mixture means are co-linear,
or the mixture means are nearly linearly dependent, the selectivity might be quite poor, in the
sense that $\frac{|E[\dotprod{\vect d_k}{ \vect m}]|}{\| \vect{d}_k \| \|\vect m\|} \ll 1$ for all $k$, including the mixture it is supposed to be selective for.
For an intuition of the geometry of this situation, see Figure \ref{BCM Geometry}.  

\begin{figure}[p!]
\centering
\includegraphics[width=.4\textwidth]{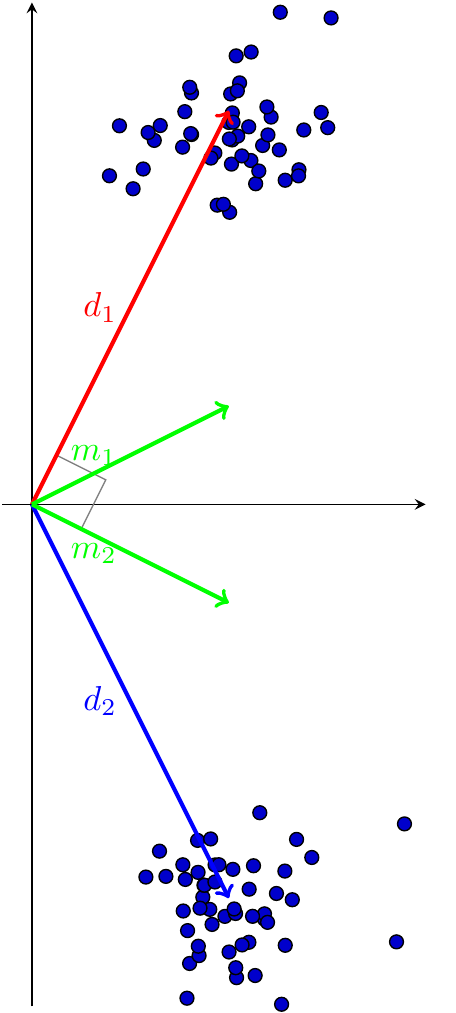}
\includegraphics[width=.55\textwidth]{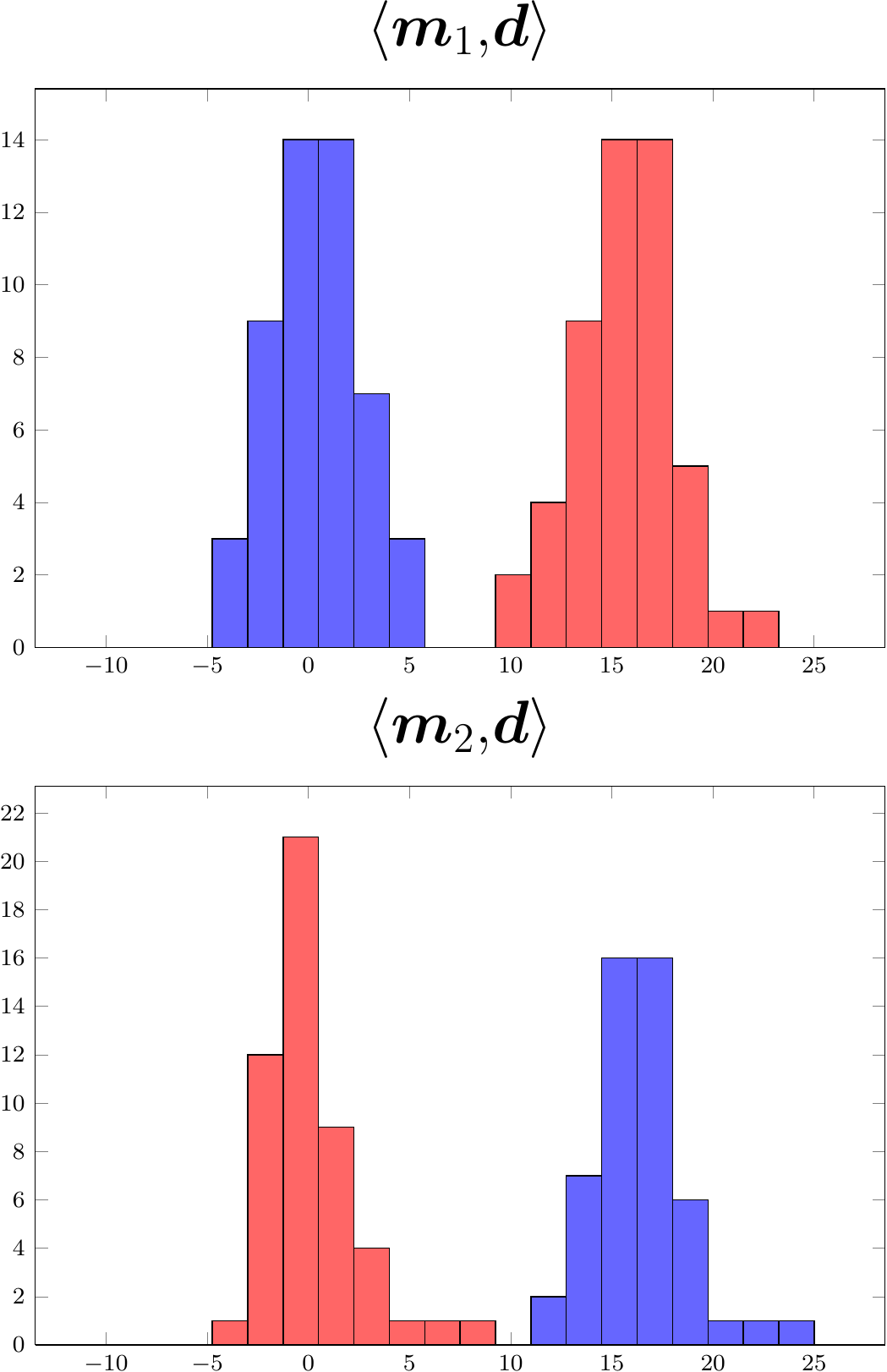}
\caption[Geometry of stable states]{Geometry of stable states.  Each $\vect{m}$ stable state is orthogonal to 
the expectation of all but one of the input distributions.} \label{BCM Geometry}
\end{figure}

We emphasize the extremely limited restrictions on the conditional distributions $P_k$.  They 
are required only to have linearly independent means, bounded variances, and the number of 
classes $K$ must be less than or equal to the dimension $N$.  Under the multi-view assumption triplet BCM converges to the same fixed point regardless of the noise distribution.
Under the multi-view assumption triplet BCM converges to the same fixed point regardless of the noise distribution.  We add that it is often possible to take 
a set of conditional distributions that do not have these properties, and add non-linear 
transformations of their dimensions as additional variables.  If the original distribution
was in fact a mixture, the transformed distribution will remain a mixture.  The transformed
version may then have the required properties.  

This suggests that this learning rule, combined with non-linear transformations, may be a 
powerful building block for learning with slowly varying data.  Examples of useful transformations include binning, bounding, and thresholding.  High dimensional 
histogram estimators can be easily constructed, either in the original signal space, or 
perhaps more plausibly in the Fourier domain to produce Gabor-like filters.  
We will not investigate the range of potential useful transformations, needless to say there
is a rich set of possible directions for future research. 

\section{Stochastic Approximation} \label{Stochastic Approximation}
Having found the stable points of the \emph{expected} update for BCM and triplet BCM, we now turn to a proof of convergence for the noisy update generated in practice.  For this, we 
turn to results from the theory of stochastic approximation.  

We will decompose our 
update into two parts, the expected update, and the (random) deviation.  This deviation will be a $L_2$ bounded martingale, while the expected update will be a ODE with the previously calculated stable points.  Since the expected update 
is the gradient of a objective function $R$, the Lyapunov functions required for the stability analysis are simply this objective function.    

The decomposition of the triplet BCM stochastic process is as follows:

\begin{align*}
\vect{m}_n - \vect{m}_{n-1} &= \gamma_n \phi(c^2_n, c^3_n, \theta_{n-1})\vect{d}^1 \nonumber \\
&= \gamma_n E[\phi(c^2, c^3, \theta_{n-1})\vect d^1] + \gamma_n \left (\phi(c^2, c^3, \theta_{n-1})\vect d^1 - E[\phi(c^2, c^3, \theta_{n-1})\vect d^1] \right) \\
&= \gamma_n h(\vect{m}_n) - \gamma_n \eta_n
\end{align*}
Here, $h(\vect{m}_n)$ is the deterministic expected update, and $\eta_n$ is a martingale.  All our
expectations are taken with respect to triplets of input data.  The decomposition for classical 
BCM is similar.  

This is the Doob decomposition \cite{doob1953stochastic} of the sequence.  Using a theorem of Delyon \cite{delyon1996general}, 
we will show that several variants of our triplet BCM algorithm will converge with probability 1, though they will require some slight modifications to guarantee convergence.   In particular, 
the unconstrained versions of our algorithm may require a finite number of projections down 
to a feasible space where the solutions actually lie.  This is due to the fact that the 
algorithm may oscillate increasingly wildly, with $\|\vect{m}_i\| \rightarrow \infty$.  

The behavior of the stochastic algorithm around stable points is intuitively clear.  The stable points act as sinks.  In some region around the stable points, the stochastic algorithm behaves
like a biased random walk, with the bias attracting the process toward the stable point.  

\begin{defin} (Delyon 1996)

Recall our update \[\vect m_n - \vect m_{n-1}= \gamma_n h(\vect{m}_n) - \gamma_n \eta_n\]
Let $\gamma_n$ be a sequence with $\sum_{i=0}^\infty \gamma_i = \infty$ and $\sum_{i=0}^\infty \gamma_i^2 < \infty$.  Let $\eta_n$ be a perturbation, $\eta_n = e_n + r_n$. 
A stochastic algorithm is A-stable if $\vect{m}_n \in K_0$ infinitely often, and the series 
$\sum \gamma_n e_n $ or $\sum \gamma_n e_n \mathbb{1}_{V(\vect{m}_n) \le M}$ converges for all $M$ and $r_n \rightarrow 0$.
\end{defin}

As is typically the case in the study of stochastic algorithms, the requirement that the 
update returns infinitely often to a compact set is quite difficult to check.  We instead
project our weights down to a more reasonable compact set where we know the true parameters lie if they ever become unreasonably large.  We note that this set can be made arbitrarily 
large, and for a sufficiently small initial step size we have found this projection does not need to be done in practice.  We note that biological neurons also have a limits
on their firing that limit the selectivity of a neuron in practice.

\begin{thm}(Delyon 1996)
	The vector field $h$ is defined on an open set $\mathcal{O} \subset \R$.  There exists a 
	nonnegative $C_1$ Lyapunov function $V$ and a finite set $\mathcal{K} \subset \mathcal O$ s.t.

	\begin{enumerate}[1)]
		\item $V(x)$ tends to $+\infty$ if $x \rightarrow \partial \mathcal{O}$ or $|x|\rightarrow \infty$
		\item $h$ is continuous and $\dotprod{\nabla V(x)}{h(x)} < 0$ if $x \notin \mathcal K$ 
		\item (Optional Projection) Let $\pi(x)$ be a continuous projection onto 
		a compact set $\mathcal{Q} \subset \mathcal{O}$ s.t. $\pi(x) = x$ for $x \in \mathcal{Q}$, and $\dotprod{\nabla V(x)}{\pi(x)- x} < -\delta |\pi(x) - x| $ for
		some $x$ in $\mathcal{O} \backslash \mathcal{Q}$
	\end{enumerate}
	Let $\vect{m}_n = \vect{m}_{n-1} + \gamma_n h(\vect{m}_{n-1}) + \gamma_n \eta_n$
	We further require that the stochastic algorithm is A-stable.  Then, $d(\vect{m}_n, \mathcal{K})$ converges to $0$. \label{Stochastic Convergence Proof}
\end{thm}
The speed of convergence and size of the respective convergence regions will be discussed later.
We need to check these conditions for each of our algorithms.  We will find that two of our 
algorithms will need to be slightly stabilized to ensure that the sequence $\vect{m}_n$ enters
a compact region infinitely often.  

In all of our algorithms the (deterministic) objective functions $R$ will act as our Lyapunov function $V$.  

For completeness we present the proof of Theorem \ref{Stochastic Convergence Proof}
\begin{proof}
	Let 
	\begin{align*}
		\vect{m_n}' &= \vect{m_n} + \sum_{i = n+1}^\infty \gamma_i e_i \\
		\vect{\delta}_n' &= -\sum_{i = n}^\infty \gamma_i e_i 
	\end{align*}
	Then, 
	\[\vect{m}_n' = \vect{m}'_{n - 1} + \gamma_n h(\vect{m}_{n - 1} + \delta_n') + \gamma_n r_n\]
	Since by assumption our
	sequence remains in a compact set, say $\mathcal C$, and our step sizes are bounded in $L_2$, by the martingale convergence theorem, $\sum_i \gamma_i e_i$ converges. By continuity of
	$h$ and since $V$ is $C_1$ we have
	\begin{align*}
	V(\vect{m}_n') &= V(\vect{m}_{n-1}') +\gamma_n h(\vect{m}_{n - 1}' + \delta_n') + \gamma_n r_n \\
	&= V(\vect{m}_{n-1}') + \gamma_n \dotprod{\nabla V(\vect{m}'_{n-1})}{h(\vect{m}_{n-1}' + \delta_n')} + \gamma_n r_n + O(\gamma_n^2)\\
	&= V(\vect{m}_{n-1}') + \gamma_n \dotprod{\nabla V(\vect{m}'_{n-1} + \delta_n')}{h(\vect{m}_{n-1}' + \delta_n')} + \gamma_n r_n' + O(\gamma_n^2)
	\end{align*}
	where $r'_n$ has absorbed the error in the dot product from shifting the Taylor series
	slightly.  This goes to zero since $\delta'_n $ goes to zero.

	Fix $\mathcal N$, an open neighborhood of the set $\mathcal K \cap \mathcal C$.  Since 
	$\mathcal C \backslash \mathcal N$ is compact, and $\dotprod{\nabla V(\vect m)}{h(\vect m)} < 0 $ 
	outside of $\mathcal K$, $\dotprod{\nabla V(\vect m)}{h(\vect m)} < -\epsilon'$ outside of 
	$\mathcal N$.  

	Since $r_n' $ goes to zero, there exists an $N$ such that $n > N$ implies
	\[V(\vect m_{n}') \le V(\vect m_{n-1}') - \gamma_n \epsilon \1_{\mathcal N^c \cap \mathcal C} + \gamma_n C \1_{\mathcal N}(\vect m_{n-1}')\]

	Let $A_\alpha$ be the set 
	\[A_\alpha = \{ x \in \R : d(x, V(\mathcal S \cap \mathcal C)) < \alpha\]
	and set $\alpha$ small enough so that $A_\alpha$ is simply disjoint intervals of 
	size $2 \alpha$, one for each zero of $V$.  Let $\mathcal N = V^{-1}(A_\alpha)$.  Then $u_n = V(\theta'_n)$ satisfies
	\[u_n \le u_{n-1} - \gamma_n \epsilon + \gamma_n C' \1_{A_\alpha}(u_{n-1})\]
	Whenever $u_n$ is out of $A_\alpha$ it decreases by at least $\gamma_n \epsilon$.  Since
	$\sum \gamma_n $ is infinite, and $u_n$ is lower bounded, whenever $u_n$ leaves $A_\alpha$ it must reach another interval of $A_\alpha$ corresponding to smaller values of $u$.  If $n$ is
	large enough such that $\gamma_n C$ is smaller than the distance between disjoint 
	intervals in $A_\alpha$ then $u_n$ cannot jump more than $\gamma_nC$.  Therefore
	$d(u_n, A_\alpha)$ must go to zero.  However, $\alpha$ was arbitrary.  So $u_n $ must
	converge to $V(\mathcal K)$.
	Since $u_n$ converges, for all $\tau$ there exists a $N(\tau)$ s.t. for all 
	$N(\tau) < n < p$ $|u_n - u_p | < \tau$, which implies that
	\begin{align*}
		\sum_{i = n + 1}^p (u_{n-1} - u_n) &< \tau \\
		\sum_{i = n + 1}^p \gamma_i \epsilon - \sum_{i = n + 1}^p \gamma_n C' \1_{\mathcal N} (\vect m_{i-1}') &< \tau
	\end{align*}
	For $p$ sufficiently large, $\sum_{i = n + 1}^p \gamma_i \epsilon > \tau$ so 
	at least one $i$ between $n$ and $p$ must have been in $\mathcal N$.  But $\mathcal N$
	was arbitrary, as was $\tau$, so $\vect m'_n$ must converge to $\mathcal K$, and therefore
	so must $\vect m_n$.
\end{proof}
\subsection{Case 1: Full Rank} \label{Full Rank proof}
We start with the simplest case.  Assume $K = N$, so the matrix of conditional expectations
 $D$ is full rank.  We further fix a large ball \[B_r := \{\vect{m}: \vect{m}^T M \vect{m} < r\}\] and a projection \[ \displaystyle \pi(\vect{m}) = 
 \begin{cases} 
\vect{m} & \vect{m}^T M \vect{m} <= r\\
r\frac{\vect{m}}{\sqrt{\vect{m}^T M \vect{m}}} &  \vect{m}^T M \vect{m} > r
 \end{cases}
\]
Let $\mathcal{O} = \R^N$
\begin{thm}
	For the full rank case, the projected update converges w.p. 1 to the zeros of
	$\nabla \Phi$
\end{thm}
\begin{proof}
Let $\mathcal{O}$ be an open neighborhood of $B$.  We replace our update 
 with its projected version 
\begin{equation}
	\vect{m} = \pi(\gamma_n \phi(c^2, c^3, \theta_{n-1})\vect{d}^1)
\end{equation}
This projection gives us the first part of the A-stability immediately.  Furthermore,
the bounded variance of each $P_k$ and the boundedness of $\vect{m}$ means 
each $c$ has bounded variance, so the martingale increment has bounded variance.  This,
plus the requirement that $\sum \gamma_i^2 < \infty$ means the martingale is bounded 
in $L_2$ so it converges.  This gives us the A-stability of the sequence.

Let $V = -R$ then conditions 1) and 2) of Delyon are clearly satisfied.  
The optional projection requirement is satisfied by noting that for some $C$
\[\frac{1}{C} \vect{m}^T M \vect{m} < \|\vect{m}\|^2 < C \vect{m}^T M \vect{m}\]
and for large enough $\vect m$
\begin{align*}
\dotprod{\nabla \Phi}{\pi(\vect{m}) - \vect{m}} < C\|\vect{m}\|^4)) \\
\text{ and } \| \pi(\vect{m} - \vect{m}) \|  &= C'(O(\|\vect{m}\|))
\end{align*}
where $C' = \frac{r}{\vect{m}^T \vect{m}} - 1$ so for sufficiently large $r$ the optional projection requirement is satisfied.  Therefore the stochastic algorithm 
converges with probability 1 to the zeros of $\nabla R$.
\end{proof}
We note that the stability of the zeros was investigated in section \ref{BCMStablePoints}
\subsection{Case 2: Low-Rank} \label{Low Rank Proof}
The case $K < N$ is somewhat trickier.  As with the full rank case, we require a projection
onto a feasible set which contains all of the stable points.  As $M$ is no longer full-rank,
we instead project using the norm $\vect{m}^t (M + \proj_{W^c})\vect{m}$, where as before, 
$W = \mspan \{\vect{d_1}, \ldots, \vect{d}_K\}$.  The expected update always lies in 
$W$, however, the martingale increment does not.  Therefore we expect convergence to 
the stable points in $W$, however we expect $\vect{m}$ to drift randomly in 
$W^\bot$.  While this does not affect the expected selectivity of the algorithm, 
it is undesirable for selectivity of the neuron to drift randomly orthogonal to subspace spanned by the true 
conditional means.  

To address this issue, we add a slight shrinkage bias to the weights. 
While a static bias would change the fixed points of the algorithm, a slowly decreasing
increment can be chosen to guarantee convergence to the stable points of 
the expected update in $W$, and to zero in $W^\bot$.  Our modified update rule will be

\begin{align}
	\vect{m}_n = \vect{m}_{n - 1} + 
	\gamma_n (-\delta_n \vect{m}_{n-1} + \phi(c_n^1, c_n^2, \theta_{n-1}) \vect{d}_n^3) \label{low rank update}
\end{align}
We assume $\gamma_n$ and $\delta_n$ have the following properties:
\begin{enumerate}
	\item $\sum_n \gamma_n \rightarrow \infty$
	\item $\sum_n \gamma_n^2 < \infty$
	\item $\delta_n \rightarrow 0$
	\item $\sum_n \gamma_n \delta_n \rightarrow \infty$
	\item $\sum_n \gamma_n^2 \delta_n^2 < \infty$
\end{enumerate}
For example, $\gamma_n = n^{-(1 - \epsilon)}$ and $\delta_n = n^{-\epsilon}$ for $0<\epsilon<\frac12$ works.
As before, we denote the expected update by $h(\vect{m})$.  We note the following facts:
$h(\vect{m}) = h(\proj_W \vect{m})$ and $R(\vect{m}) = R(\proj_W \vect{m})$.  We split our process into two processes.
\begin{align}
\vect{\varphi}_n &= \pi_1(\proj_W \vect{m}_n) \nonumber \\
		&= \pi_1(\varphi_{n-1} + \gamma_n(-\delta_n \vect{\varphi}_{n-1} + h(\vect{\varphi_{n-1}}) + 
		\eta_n')) \\
\vect{\psi}_n &= \pi_2(\proj_{W^\bot} \vect{m}_n) \\
			&= \pi_2(-\gamma_n (\delta_n \vect{\psi}_{n-1} + \eta_n''))
\end{align}
where $\eta_n' = \proj_W{\eta_n}$ and $\eta_n'' = \proj_{W\bot} \eta_n$.  We note that
$\eta_n'$ and $\eta_n''$ are $\mathcal{F}_n$-measurable martingales, where 
$\mathcal{F}_n$ is the sequence of $\sigma$-algebras generated by the sequence $\vect{m}_n$. 

 First, we show $\vect{\varphi}_n$ converges to one of the stable states of the expected update.  
 note that $-\delta_n \vect{\varphi}_{n_1} \rightarrow 0$ so it meets the definition of $r_n$ in the
 definition of A-stability.  The expected update and martingale increment of $\vect{\varphi}$ behaves precisely like the 
 full rank case, except for some extra variability in the martingale increment which remains 
 controlled by the projections $\pi_1, \pi_2$.  As before, we use $R$ as our Lyapunov function, noting that $R(\vect{m})$ depends only on the $\vect \varphi$ component of $m$.  As the column space of $M$ is $W$, our
 previous projection only restricts $\vect \varphi$ to a compact space.  We use it as our 
 $\pi_1$.  
 
 For $\vect{\psi}$, the expected update $-\gamma_n \delta_n \vect \psi_{n-1}$ has only one fixed point at
 zero.  $\pi_2$ can be any projection onto a large ball in $W^\bot$, say $\|\vect \psi \| < r$.  For our Lyapunov function,
 $\|\vect{\psi}\|^2 $ trivially satisfies all of the required conditions.  Though the step size
 decays more rapidly than for $\vect \varphi$, it does not decay rapidly enough for 
 $\sum_n \gamma_n \delta_n$ to converge.

 In practice, we often care more about $M \vect m$ than $m$ itself, as that is directly 
 comparable to the parameters of the mixture model.  Since the column space of $M$ is $W$, 
 this shrinkage can be safely ignored, at the expense of increased variance of the 
 martingale increment, and stochastic drift of the orthogonal component.   

\section{Networks of BCM Neurons}\label{Network}
Like the classical BCM neuron, the modified BCM neuron can be arranged in a network with lateral inhibition or excitation.  This network will not affect the stable points of individual neurons, each will remain selective for one and only one mixture component in expectation.  However it will affect the distribution of stable points across neurons.  While proving this, we will 
correct a few errors in the proof of the original Intrator result, which is the basis of this section, and give a cleaner characterization of the network using the Kronecker product.  This characterization will decouple
the network activity from the selectivity of individual neurons.
\subsection{Kronecker Product}
The Kronecker product gives a canonical matrix form for tensor products.  We will denote the Kronecker product by $\ktimes$.  
\begin{align*}
A_{m \times n} = \begin{bmatrix}
	a_{11}  & \cdots & a_{1n} \\
	\vdots & \ddots & \vdots \\
	a_{m1} & \cdots & a_{mn}
\end{bmatrix} & B_{p \times q} = \begin{bmatrix}
	b_{11}  & \cdots & b_{1q} \\
	\vdots & \ddots & \vdots \\
	b_{p1} & \cdots & b_{pq}
\end{bmatrix} \\
(A\ktimes B)_{mp \times nq} &= \begin{bmatrix}
	a_{11}B  & \cdots & a_{1n}B \\
	\vdots & \ddots & \vdots \\
	a_{m1}B & \cdots & a_{mn}B
\end{bmatrix}
\end{align*}

We use only a few facts about Kronecker products.
\begin{align}
(A \ktimes B) (C \ktimes D) &= AC \ktimes BD \\
(A \ktimes B)^{-1} &= (A^{-1} \ktimes B^{-1}) \\
(A \ktimes B)^T &= A^T \ktimes B^T \\
I_n \ktimes I_m &= I_{nm} \\
\text{if } A\vect{v} = \lambda \vect{v} &\text{ and } B\vect{w} = \mu \vect{w} \text{ then } \nonumber \\(A \ktimes B) (\vect{v} \ktimes \vect{w}) &= \lambda \mu \vect{v} \ktimes \vect{w}
\end{align}
In particular, if $A$ and $B$ are p.s.d., so is $A \ktimes B$.  All of these properties follow
trivially from the definition.  

\subsection{Networks of BCM Neurons (Intrator 1996)} \label{Networks of BCM}
With the notation of Kronecker products settled, we present a proof that the expected update
for triplet BCM with network interactions has the same stable states as does triplet BCM.  That
is to say, each neuron will be selective for a single mixture component in expectation.  Numerical results show that the network interactions can be modulated to ensure that the 
neurons converge to the same or different states through excitatory or inhibitory lateral networks.  
\begin{figure}[ht]
	\centering
	\includegraphics[width=.75\textwidth]{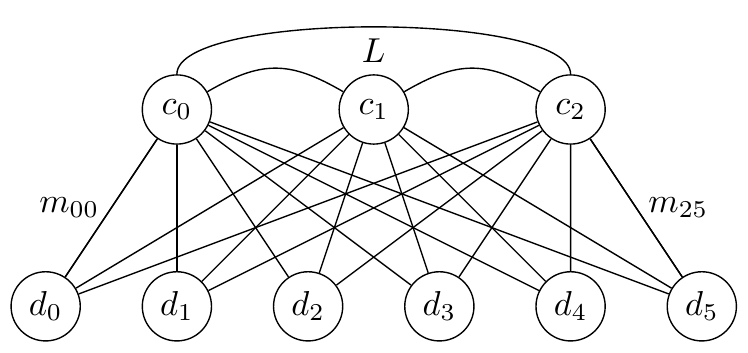}
	\caption[Configuration for a network of BCM neurons]{Configuration for a network of BCM neurons.  All neurons receive the
	same input, and are connected with a network of lateral connections, $L$\label{fig:lateral}}
\end{figure}
We begin with a network of $n$ neurons.
Let $\vect{m}_i$ denote the vector of synaptic weights for neuron $i$.  Let $c_{ik} = E_{P_k}[\dotprod{\vect{m}_i}{\vect{d}}]$, the expected response of neuron $i$ to class $k$.  Let $\Phi_{ik} = \phi(c_{ik}, c_{ik}, \theta) = E_{P_k} \phi(c_i^2, c_i^3, \theta)$.  Let $P$ be the diagonal matrix of class probabilities, $P_{ii} = \alpha_i$.  Finally, let $D$ be the matrix of expected data vectors, such that $D_{ij}$ is the $j$th entry of the mean of distribution $i$.  

To organize the computation over multiple neurons, we use the Kronecker product.  Let $\mathcal{D} = I_n \ktimes D$ and $\mathcal{P} = I_n \ktimes P$.
\begin{align*}
\mathcal{D} &= \begin{bmatrix}D &\cdots &0 \\ \vdots &\ddots &\vdots \\0 & \cdots& D   \end{bmatrix}& \mathcal{P} = \begin{bmatrix}P &\cdots &0 \\ \vdots &\ddots &\vdots \\0 & \cdots& P   \end{bmatrix}
\end{align*}
We set the multivectors $\vect{c} = [c_{11},  c_{12}, \ldots, c_{nk}]^T$ and $\vect{m} = [m_{11}, m_{12}, \ldots, m_{nd}]^T$.

With this notation, the expected per class firing rate for a neuron $\vect c$ for the triplet 
BCM rule can be written as 
\[\vect c  = \mathcal D \vect m\]
We will modify this rule to incorporate lateral interaction between all neurons, as seen in 
figure \ref{fig:lateral}.  We assume that the firing rate of neuron $\vect c_{i\bullet}$ is a
linear combination of its firing rate due to its input, and the firing rate of all other 
neurons in the network.  We denote the matrix of these weights by $L$.
We assume (unfortunately) that the lateral connections between neurons occur essentially instantaneously, so that the firing rate of the neuron equilibrates instantly.  Therefore
\[\vect c_{i j} = D \vect m_{i \bullet} + L \vect c_{\bullet j}\]
The matrix $L$ is a symmetric connection matrix and we assume the operator norm of $L$, $|L| < 1$.  We set $\mathcal{L} = L \ktimes I$.  Note the difference in the position of the identity matrix between $\mathcal{D}$ and $\mathcal{L}$.  The left hand side of the Kroecker product represents signals applied uniformly within each neuron, while the right hand side represents signals applied across neurons.  With this notation and assumptions, 
\begin{align}
\vect{c} &= \mathcal{D} \vect{m} + \mathcal{L}\vect{c} \\
\vect{c} &= (I - \mathcal{L})^{-1} \mathcal{D}\vect{m} \nonumber \\
&= ((I - L)^{-1} \ktimes D) \vect{m} 
\end{align}

Conveniently, the network interactions and feedforward input remain decoupled, with 
the network interactions on the left side of the Kronecker product, and the neural 
selectivity on the right side.   
We will show this decoupling prevents the 
lateral connections from affecting the stability of the stable states.  The expected update is then

\begin{equation}
E[\dot{\vect{m}}] = \mathcal{D}^T \mathcal{P} \Phi \label{Network expected update}
\end{equation}

where the conditional independence of the independent samples have been used both within $\Phi$ and in the product.  
\begin{thm}
If $|L| < 1$, and all of the input criteria of Theorem \ref{Full Rank proof} are met, then
the expected update is stable (in the sense of Lyapunov) when each neuron is selective for
one and only one mixture component.  It is possible that multiple neurons are selective for
the same mixture component.  
\end{thm}
\begin{proof}
This expected update is zero iff $\Phi = 0$.  Note that $\Phi$ depends on $\vect m$ only through the vector of expected firing rates $\vect c$.  Our calculation for the zeros of the expected update in the proof of Theorem \ref{BCMStablePoints} only depends on $c$, so the critical points remain unchanged.  However, it is possible that the stability of those solutions will be affected
by the lateral connections.

We need to show that the Jacobian of the expected update is positive semidefinite only when the
neuron is selective for one state.  
\begin{align*}
J E[\dot{\vect{m}}] &= \mathcal{D}^T \mathcal{P} \left(\frac{\partial \Phi}{\partial \vect{c}}\right) \frac{\partial c}{\partial \vect{m}} \\
&= \mathcal{D}^T \mathcal{P} \left(\frac{\partial \Phi}{\partial \vect{c}}\right) (I - \mathcal{L})^{-1}\mathcal{D}\\
&= \left [I \ktimes (D^T P)\right ] \left(\frac{\partial \Phi}{\partial \vect{c}}\right) \left [(I - L)^{-1} \ktimes D \right ]
\end{align*}
$\frac{\partial \Phi}{\partial \vect{c}}$ can be decomposed into individual neurons,
\[ \frac{\partial \Phi}{\partial \vect{c}} = \sum_{i=1}^n I_{\{i\}} \ktimes \frac{\partial \Phi_i}{\partial c_i}\]
so 
\begin{align}
 J E[\dot{\vect{m}}] =\sum_i I_{\{i\}}(I - L)^{-1} \ktimes (D^T P \frac{\partial \Phi_i}{\partial c_i} D)
\end{align}
this can only be stable if each $D^T P \frac{\partial \Phi_i}{\partial c_i}D$ is stable, which by the previous 
analysis occurs only when $c$ is selective for only one state.  For this case, $\frac{\partial \Phi_i}{\partial c_i} = \frac{1}{\alpha_i}I$ where $\alpha_i$ is the probability of the 
state neuron $i$ is selective for, once again by the proof of Theorem \ref{BCMStablePoints}.  Therefore 
\begin{equation}
	J E[\dot{\vect{m}}] = \left[\operatorname{diag}\left(\frac{1}{\alpha_i}\right) (I - L)^{-1} \right]\ktimes[  D^T P D]
\end{equation}
Each part of the Kronecker product is positive semidefinite, so the expected update is positive semidefinite.  
Thus, the expected update is stable, and each neuron is selective in expectation for one and only one member
of the class.  
\end{proof}

The resulting weights are given by $\vect{m} = D^{-1}(I - \mathcal{L})\vect{c}$
This network allows for a parallel search for the mixture components.  We simply feed 
the same inputs to a collection of triplet BCM neurons, and connect them with an inhibitory 
network.  The inhibitory network will drive the probability of them converging to the same input down, 
which gives an approximate parallel search of the parameter space.  
We demonstrate this parallel search with the following two dimensional example.
\begin{experiment}
Let $\vect d$ be distributed according to a Gaussian mixture model with two components.  
\[P(\vect d) = \sum_{i=1}^2 A \alpha_i \exp\left(- \frac{\| \vect d - \vect e_i\|^2}{2\sigma}\right) \]
where $\vect e_i$ is a standard basis vector.
$\vect{\alpha} = [0.4, 0.6]$, and $\sigma = 10^{-4}$.  We build a network of two neurons 
with a lateral matrix \[L = a \begin{bmatrix} 0 & 1 \\ 1 & 0\end{bmatrix}\]
Each neuron modified its weights according to the triplet BCM rule, and independent triples
were drawn from the mixture model.  
We randomly initialize the weights on the unit ball, and measure the number of times the neurons converge to be selective for different mixture states.  We used 30 different random initializations per $a$.  

\begin{table}[h]
\centering
\begin{tabular}{r|l}
$a$ & \% different \\
\hline
-.25  &  83.3\\
-.125 & 53.3 \\
 0    & 46.7 \\
.125  & 26.7 \\
.25   & 3.3 \\
\end{tabular}
\caption[Probability that a network of triplet BCM neurons will converge to different
mixtures]{Number of times a network of two triplet BCM neurons converge to the same 
state out of twenty repetitions.  The lateral network was set at $L=a\begin{bmatrix} 0 & 1 \\ 1 & 0 \end{bmatrix}$}
\end{table}
Even a modest excitatory network causes the neurons to converge to the same state with high 
probability.  An inhibitory network has a similar but weaker effect in the opposite direction,
encouraging the neurons to become selective for different mixture components.  

\end{experiment}

\section{Conclusion}
We presented a novel learning rule we call triplet BCM.  We proved that under a multi-view assumption and input from a restricted class of mixture models, this learning rule provably learns selectivity to one mixture mean.  The only restriction on the mixture models are that the mixture means are linearly independent, and that they have bounded variance.  Furthermore, this learning rule can be trivially implemented neurally without any feedback mechanism, and only a sliding threshold needs to be maintained per neuron.  We also demonstrated that networks of triplet BCM neurons can be combined with a lateral network to force each neuron to learn a different component of the mixture model.  

We believe the connection between classical BCM and tensor decomposition provides new insights into the information processing role of neural circuits.  A future publication will illustrate the connection between this work and synaptic modification through spike timing dependent plasticity, an important learning mechanism in cortex.  

\bibliographystyle{plain}
\bibliography{bibliography}

\end{document}